\documentclass[a4paper,UKenglish]{lipics-v2018}

\usepackage{microtype}

\usepackage{float}
\usepackage{thm-restate}
\usepackage{graphicx}

\newcommand{\link}{{\tt link}}
\newcommand{\dm}{\mbox{{\tt delete}-{\tt min}}}
\newcommand{\dc}{\mbox{{\tt decrease}-{\tt key}}}
\newcommand{\ins}{\mbox{{\tt insert}}}
\newcommand{\meld}{\mbox{{\tt meld}}}

\hideLIPIcs

\title{Pairing heaps: the forward variant}

\author{Dani Dorfman}{Blavatnik School of Computer Science, Tel Aviv University, Israel}{dannatand@mail.tau.ac.il}{}{}

\author{Haim Kaplan}{Blavatnik School of Computer Science, Tel Aviv University, Israel}{haimk@post.tau.ac.il}{}{Research supported by The Israeli Centers of Research Excellence (I-CORE) program (Center No.\ 4/11), Israel Science Foundation grant no.\ 1841-14.}

\author{L\'{a}szl\'{o} Kozma}{Eindhoven University of Technology, The Netherlands}{l.kozma@tue.nl}{}{Work done while at Tel Aviv University. Research supported by The Israeli Centers of Research Excellence (I-CORE) program (Center
No.\ 4/11).}

\author{Uri Zwick}{Blavatnik School of Computer Science, Tel Aviv University, Israel}{zwick@tau.ac.il}{}{Research supported by BSF grant no.\ 2012338 and by
The Israeli Centers of Research Excellence (I-CORE) program (Center
No.\ 4/11).}

\authorrunning{D. Dorfman, H. Kaplan, L. Kozma and U. Zwick}
\Copyright{Dani Dorfman, Haim Kaplan, L\'{a}szl\'{o} Kozma, Uri Zwick}


\EventShortTitle{arXiv}
\EventAcronym{arXiv}
\nolinenumbers 

\begin{document}

\maketitle

\begin{abstract}
The pairing heap is a classical heap data structure introduced in 1986 by Fredman, Sedgewick, Sleator, and Tarjan. It is remarkable both for its simplicity and for its excellent performance in practice. The ``magic'' of pairing heaps lies in the restructuring that happens after the deletion of the smallest item. The resulting collection of trees is consolidated in two rounds: a left-to-right pairing round, followed by a right-to-left accumulation round. Fredman et al.\ showed, via an elegant correspondence to splay trees, that in a pairing heap of size $n$ all heap operations take $O(\log{n})$ amortized time. They also proposed an arguably more natural variant, where both pairing and accumulation are performed in a combined left-to-right round (called the forward variant of pairing heaps). The analogy to splaying breaks down in this case, and the analysis of the forward variant was left open. 

In this paper we show that inserting an item and deleting the minimum in a forward-variant pairing heap both take amortized time $O(\log{n} \cdot 4^{\sqrt{\log n}})$. This is the first improvement over the $O(\sqrt{n})$ bound showed by Fredman et al.\ three decades ago. Our analysis relies on a new potential function that tracks parent-child rank-differences in the heap.

 \end{abstract}

\section{Introduction}

A heap is an abstract data structure that stores a set of keys, supporting the usual operations of creating an empty heap, inserting a key, finding and deleting the smallest key, decreasing a key, and merging (``melding'') two heaps. Heaps are ubiquitous in computer science, used, for instance, in Dijkstra's shortest path algorithm and in the Jarn\'{i}k-Prim minimum spanning tree algorithm (see e.g.,~\cite{Cormen}). Heaps are typically implemented as trees (binary or multiary), where each node stores a key, and no key may be smaller than its parent-key.

The pairing heap, introduced by Fredman, Sedgewick, Sleator, and Tarjan~\cite{pairing}, is a popular heap implementation that supports all operations in $O(\log{n})$ amortized time. It is intended to be a simpler, \emph{self-adjusting} alternative to the Fibonacci heap~\cite{fibonacci}, and has been observed to have excellent performance in practice~\cite{Stasko,Moret,Katriel,Tarjan14}. Despite significant research, the exact complexity of pairing heap operations is still not fully understood~\cite{pairing,FredmanLB,IaconoUB,Pettie,Elmasry1,IaconoOzkan,Iacono16}.

\subparagraph*{Self-adjusting data structures.} Many of the fundamental heap- and tree- data structures in the literature are designed with a certain ``good structure'' in mind, which guarantees their efficiency. Binary search trees, for instance, need to be (more or less) balanced in order to support operations in $O(\log{n})$ time. In the case of heaps, a reasonable goal is to keep node-degrees low, as it is costly to delete a node with many children. Maintaining a tree structure in such a favorable state at all times requires considerable book-keeping. Indeed, data structures of this flavor (AVL-trees~\cite{avl}, red-black trees~\cite{Bayer1972, Guibas78}, Binomial heaps~\cite{Vuillemin78}, Fibonacci heaps~\cite{fibonacci}, etc.) store structural parameters in the nodes of the tree, and enforce strict rules on how the tree may evolve.

By contrast, pairing heaps and other \emph{self-adjusting} data structures store no additional information besides the keys and the pointers that represent the tree itself. Instead, they perform local re-adjustments after each operation, seemingly ignoring global structure. Due to their simplicity and low memory footprint, self-adjusting data structures are appealing in practice. Moreover, self-adjustment allows data structures to \emph{adapt} to various usage patterns. This has led, in the case of search trees, to a rich theory of instance-specific bounds (see e.g.,~\cite{ST85, in_pursuit}). 
The price to pay for the simplicity and power of self-adjusting data structures is the complexity of their analysis. With no rigid structure, self-adjusting data structures are typically analysed via potential functions that measure, informally speaking, how far the data structure is from an ideal state.

The standard analysis of pairing heaps~\cite{pairing} borrows the potential function developed for splay trees by Sleator and Tarjan~\cite{ST85}. This technique is not directly applicable to other variants of pairing heaps. There is a large design space of self-adjusting heaps similar to pairing heaps, yet, we currently lack the tools to analyse their behavior (apart from isolated cases). We find it a worthy task to develop tools to remedy this situation.

\subparagraph{Description of pairing heaps.}
A pairing heap is built as a single tree with nodes of arbitrary degree. Each node is identified with a key, with the usual min-heap condition: every (non-root) key is larger than its parent key. 

The heap operations are implemented using the unit-cost \emph{linking} primitive: given two nodes $x$ and $y$ (with their respective subtrees), $\link(x,y)$ compares $x$ and $y$ and lets the greater of the two become the leftmost child of the smaller.

The \dm~operation works as follows. 
We first delete (and return) the root $r$ of the heap (i.e., the minimum), whose children are 
$x_1, \dots, x_k$. 
We then perform a left-to-right pairing round, calling $\link(x_{2i-1}, x_{2i})$ for all $i=1, \dots, \lfloor k/2 \rfloor$. The resulting roots are denoted $y_1, \dots, y_{\lceil k/2 \rceil}$. (Observe that if $k$ is odd, the last root is not linked.) Finally, we perform a right-to-left accumulate round, calling $\link(p_{i},y_{i-1})$ for all $i=\lceil k/2 \rceil, \dots, 2$, where $p_i$ is the minimum among $y_i,\dots,y_{\lceil k/2 \rceil}$. We illustrate this process in Figure~\ref{fig1}.

The other operations are straightforward. In \dc, we cut out the node whose key is decreased (together with its subtree) and link it with the root. In \ins, we link the new node with the root. In \meld, we link the two roots.

Fredman et al.~\cite{pairing} showed that the amortized time of all operations is $O(\log{n})$. They also conjectured that, similarly to Fibonacci heaps, the amortized time of \dc~is in fact constant. This conjecture was disproved by Fredman~\cite{FredmanLB}, who showed that in certain sequences, \dc~requires $\Omega(\log\log{n})$ time. In fact, the result of Fredman holds for a more general family of heap algorithms. Iacono and \"{O}zkan~\cite{IaconoOzkan, IaconoOzkan2} gave a similar lower bound for a different generalization\footnote{Both results are rather subtle, for instance, it is not clear whether \dc~remains costly in an implementation where we cut an affected node only if its decreased key is smaller than that of its parent.} of pairing heaps. Assuming $O(\log{n})$ time for \dc, Iacono~\cite{IaconoUB} showed that \ins~takes constant amortized time. Pettie~\cite{Pettie} proved that both \dc~and \ins~take $O(4^{\sqrt{\log\log{n}}})$ time.  Improving these trade-offs remains a challenging open problem.

\subparagraph*{Pairing heap variants.}

We refer to the pairing heap data structure described above as the \emph{standard} pairing heap. Fredman et al.~\cite{pairing} also proposed a number of variants that differ in the way the heap is consolidated during a \dm~operation.

We describe first the \emph{forward variant}, which is the main subject of this paper (in the original pairing heap paper this is called the \emph{front-to-back variant}). The pairing round in the forward variant is identical to the standard variant, resulting in roots $y_1, \dots, y_t$, where $t=\lceil k/2 \rceil$, and $k$ is the number of children of the deleted root. In the forward variant, however, we perform the second round also from left to right: for all $i=1, \dots, t-1$, we call $\link(p_{i},y_{i+1})$, where $p_i$ is the minimum among $y_1,\dots,y_i$. Occasionally we refer to $p_i$ as the current \emph{left-to-right minimum}. See Figure~\ref{fig1} for an illustration. What may seem like a minor change causes the data structure to behave differently. Currently, the forward variant appears to be much more difficult to analyse than the standard variant.

The implementation of the forward variant is arguably simpler. The two passes can be performed together in a single pass, using only the standard \emph{leftmost child} and \emph{right sibling} pointers. In fact, it is possible to implement the two rounds in one pass also in the standard variant. To achieve this, we need to perform both the pairing and accumulation rounds \emph{from right to left}. As shown by Fredman et al.~\cite{pairing}, this can be done with a slightly more complicated link structure, and the original analysis goes through essentially unchanged. Regardless of the low-level details, it remains the case that of the two most natural pairing heap variants one is significantly better understood than the other. 

Yet another variant described by Fredman et al.~\cite{FredmanLB} is the \emph{multipass} pairing heap. Here, instead of an accumulation round, repeated pairing rounds are executed, until a single root remains. For multipass pairing heaps, the bound of $O\bigl(\log{n} \cdot \log\log{n} / \log\log\log{n}\bigr)$ was shown~\cite{pairing} on the cost of a \dm.\footnote{A recently claimed improvement by Pettie~\cite{pettie_slides} would reduce this to an almost, but not quite, logarithmic bound. The result of Pettie has, in some sense, inspired our results.}

Fredman et al.~\cite{pairing} also describe what we could call the \emph{arbitrary pairing and linking} variant. Here, before performing the initial pairing round, the roots may be arbitrarily reordered. After the pairing round, arbitrary pairs of roots (not necessarily neighbors) are linked, until there is a single root left. For this general case (that subsumes all previous variants) Fredman et al.~\cite{pairing} show a tight\footnote{It is not clear how efficient this strategy is if we only allow \link~operations between neighboring siblings.} $\Theta(\sqrt{n})$ amortized bound for \dm. Even in the special case of the forward variant, the $O(\sqrt{n})$ upper bound has never been improved. 

Our main result is the following.

\begin{restatable}[]{theorem}{mainclaim}\label{thm2}
In the forward variant of pairing heaps, the amortized costs of \dm~and \ins~are $O(\log{n} \cdot 4^{\sqrt{\log{n}}})$, where $n$ is the number of items in the heap when the operation is performed.
\end{restatable}

The result holds for sequences of \ins~and \dm~operations arbitrarily intermixed, starting from an empty heap. Note that an \ins~operation performs a single \link, its worst-case cost is therefore constant. 

The quantity in the running time is  asymptotically smaller than $n^\varepsilon$, for all $\varepsilon > 0$. We remark that the $4^{\sqrt{\log{n}}}$ term grows much faster than the $\log{n}$ term, as for arbitrary constants $c,d > 1$ it holds that $c^{\sqrt{\log{n}}} = \omega(\log^d{n})$.

Besides the concrete result (heaps with proven logarithmic cost are known, afterall), the contribution of our paper is in the development of a new, fairly general potential function, that may have further applications.



There has been significant further work in designing heap data structures, with the goal of finding a simpler alternative to Fibonacci heaps, matching, or almost matching their theoretical guarantees (e.g.,~\cite{Stasko, ThinThick, Violation, RankPairing, Hollow, Quake}). These heaps are typically more complicated than pairing heaps. Moreover, they are not self-adjusting, i.e., they store extra information at the nodes, are thus out of the scope of this paper.

\begin{figure}
	\begin{center}
\includegraphics[width=0.99\textwidth]{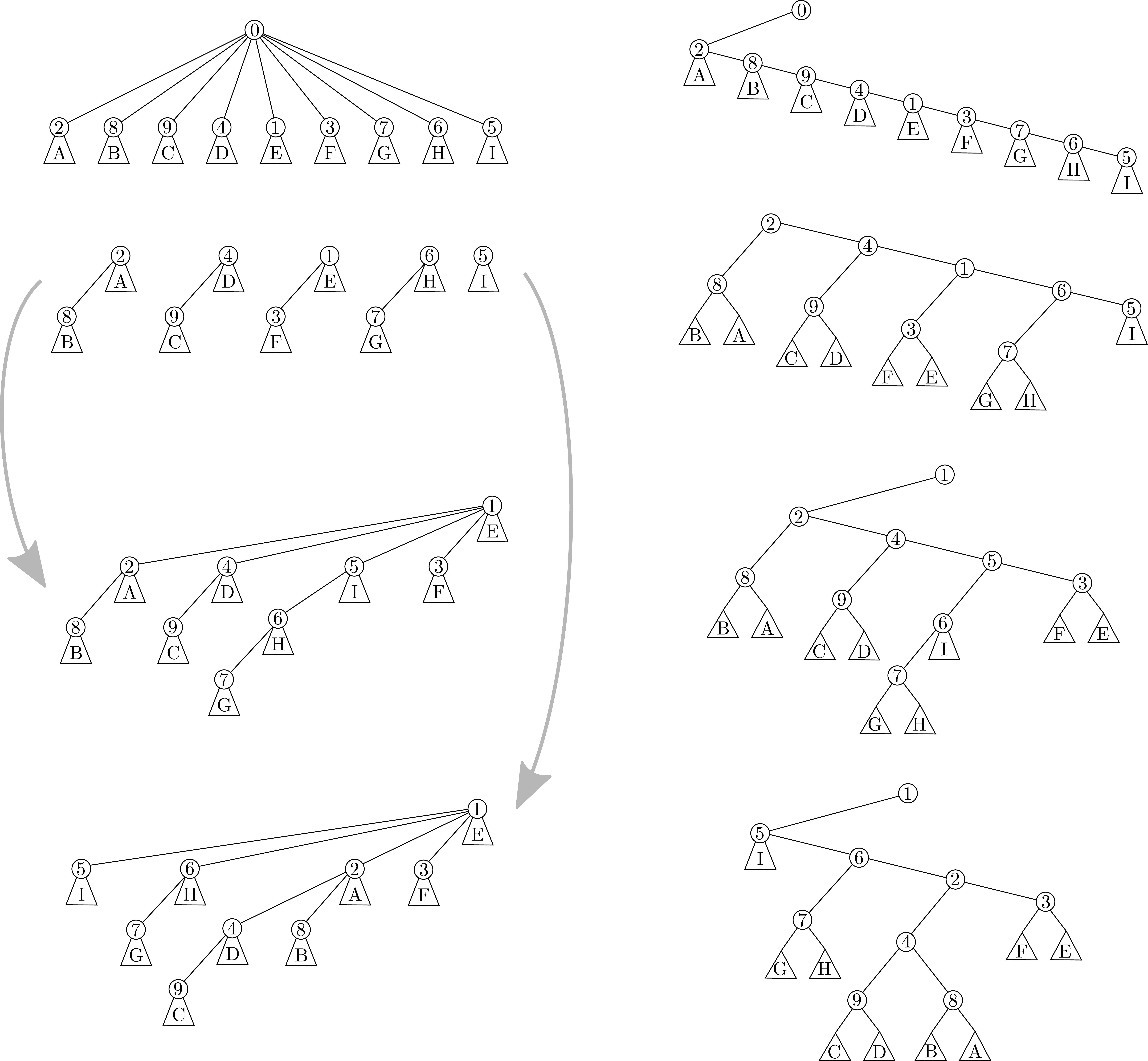}
	\end{center}
\caption{Pairing heap \dm~operation. Left: heap-view, right: binary tree-view (splay-view). From top to bottom: (i) initial heap, (ii) deleting the root and a left-to-right pairing round, (iii) standard variant, resulting heap after right-to-left accumulation round, (iv) forward-variant, resulting heap after left-to-right accumulation round. Circled numbers are keys, letters indicate arbitrary subtrees. Observe that the forward variant (iv) reverses the ordering of the siblings with keys 2,6,5. 
\label{fig1}}
\end{figure}

\vspace{-0.08in}
\subparagraph*{Splaying and sorting.}

There is a standard representation of multi-way rooted trees as binary trees (\cite[\S\,2.3.2]{Knuth}, \cite{pairing}), by renaming the \emph{leftmost child} and \emph{right sibling} pointers as \emph{left child} and \emph{right child}. (See Figure~\ref{fig1}.) In this binary tree representation, the restructuring performed when we delete the minimum from a pairing heap closely resembles the restructuring performed by an access to an item in a splay tree. The correspondence is sufficiently close that the proof of logarithmic access cost in splay trees~\cite{ST85} also goes through for the pairing heap. (We note that this correspondence is far from trivial: one has to relabel some nodes and swap their left and right children in the analysis to make it work; we refer to the original pairing heap paper~\cite{pairing} for details.)

In case of the forward variant, the analogy with splay trees breaks down, with no apparent way to fix it. The reason is that the accumulation round of the forward variant may reverse the ordering of large groups of consecutive neighbors. In splay-view, this has the effect of turning long portions of the search path upside-down, something known to be notoriously hard to analyse, see~\cite{Subramanian96, ESA15}.

There is a close link between heaps and sorting. The ancestor-descendant relationship in a heap captures everything we know about the order of the keys at a certain time. A sequence of \dm\;\!s in a pairing heap (or in any other heap, for that matter), can be seen as transitioning from the current partial order towards the total order, i.e., sorting. (More precisely, we have here a special kind of selection-sort, in which only candidate minima are ever compared.) It would then be natural for the potential function to capture some measure of ``sortedness'' of the heap, for example, the entropy of the partial order or some other quantity related to the number of linear extensions. In general, as observed by Pettie~\cite{Pettie}, this is far from being the case for the splay potential.

The potential function in the classical analysis of splay trees is equal to the sum of the logarithms of subtree-sizes of all nodes in the tree. A subtree in the splay-view corresponds, in the heap-view, to the set of subtrees of a node and all of its right siblings (Figure~\ref{fig1}). It is far from intuitive why such a quantity would be useful in analysing pairing heaps. In particular, this potential function does not distinguish between left and right children in the splay-view, even though, in the heap-view, they play different roles: one is a provably larger key, and the other is (so far) incomparable. It may seem that, implicitly, the splay potential assumes that the left-to-right ordering of siblings (in heap-view) is, to some extent, correlated with the sorted order. We argue that such an assumption is reasonable for the standard variant of pairing heaps, but not for the forward variant (since the accumulation round of the forward variant frequently reverses blocks of consecutive siblings).

As a simple example, consider either the standard or the forward variant, and look at a group of consecutive siblings $x_1, x_2, \dots, x_k, y$ (arbitrarily shuffled), of which $y$ is the smallest. Suppose, for simplicity, that these nodes interact with each other \emph{only in pairing rounds}, i.e., the minimum of the accumulation round always comes from ``outside the group'', and does not change within the group. Assume also that, in each pairing round following a deletion of the parent of the group, $y$ gains one node from the group as its new leftmost child. 

If $x_i$ and $x_j$ both become children of $y$, and $x_i$ is to the left of $x_j$, then $x_i$ became the child of $y$ in a later round than $x_j$, having ``survived'' more rounds (winning the links it took part in) at the same level as $y$. This fact is an indication that $x_i$ may be smaller than $x_j$. When $y$ is eventually deleted, the standard variant will preserve the order of its children, whereas the forward variant will reverse it (assuming again, that there is no minimum-change within the group during the accumulation round).

There are limits to this intuition, and it is easy to construct examples that break it. Nonetheless, this interpretation suggests that we analyse the forward variant by somehow directly using the order-information, instead of trying to infer it from the tree structure. This motivates our potential function in \S\,\ref{sec2}.

The intuition described earlier also hints at why the standard variant of pairing heaps may in fact be faster than the forward variant. If the left-to-right ordering of siblings is indeed the increasing order, then a right-to-left accumulate round is vastly more efficient then a left-to-right round. (The former immediately achieves the sorted order, whereas the latter merely finds the minimum.) On the other hand, in this case, the left-to-right round also reverses the order of siblings, making it optimal for the subsequent round. This intuition is consistent with our experiments that show the forward-variant to be somewhat slower\footnote{Experiments also suggest that multipass is somewhat slower than the standard variant~\cite{Stasko}.} than the standard variant. However, the possibility that the forward-variant also has logarithmic cost has not been ruled out.

It remains an interesting open question to determine the exact complexity of the forward variant, and to characterize the types of instances on which it may outperform the standard variant. We also leave open the question whether \dc~may take $o(\log{n})$ time in this variant, noting that the lower bound of Fredman~\cite{FredmanLB} applies to this case, whereas the upper bound of Pettie~\cite{Pettie} does not.

More importantly, one can hope that new techniques for the analysis of pairing heaps will find their way to the analysis of splay trees and other dynamic search tree algorithms. Splaying and its variants pose some of the most intriguing and central open questions of the field, such as the dynamic optimality conjecture~\cite{ST85, in_pursuit}.

\section{Analysis of the forward variant of pairing heaps\label{sec2}}

Before proving our main result, as a warm-up, we look at the \emph{arbitrary pairing and linking} variant of pairing heaps. First we introduce some terminology.

We define the \emph{rank} of a node $x$ as the number of nodes in the heap with a smaller key,\footnote{Contrary to many other \emph{rank}s in the data structures literature, we use the word in its ``original'' meaning.} e.g., the rank of the root is $0$. (For simplicity, we assume that the keys are unique.) We denote the rank of $x$ by $r(x)$. The \emph{rank-difference} $rd(x)$ of a node $x$ is defined as $rd(x) = r(x) - r(p(x))$, where $p(x)$ is the parent node of $x$. It is clear that $rd(x)$ is positive for all non-root nodes $x$. For the root $r$, we define $rd(r) = 0$. 

\subsection{Arbitrary pairing and linking\label{sec20}}

We show that starting with an arbitrary initial heap of size $n$, the cost of $n$ \dm\;\!s using arbitrary pairing and linking is $O(n\sqrt{n})$. This was already known, as a corollary of Fredman et al.'s result in the original pairing heap paper~\cite{pairing}. 
Our proof is different (and arguably simpler), and is intended to illustrate the use of rank-differences in the analysis. 

Consider a heap of size $n$. Let $\Phi$ denote the sum of rank-differences over all nodes of the heap. Observe that the ranks take all values $0, \dots, n-1$, and the rank-difference of a node is not more than its rank. It follows that $0 \leq \Phi \leq n(n-1)/2$.

For the purpose of the analysis, we slightly change the implementation of \dm. Rather than deleting the root first, we first do the pairing and accumulation rounds on the children of the root, until the root has only one child. Only then, we actually delete the root. This modified algorithm is equivalent to the original.

In this implementation, all \link~operations take place between children of the root. Consider a \link~between nodes $x$ and $y$, whose rank-differences are $a$ and $b$, respectively. Suppose $x < y$, and consequently $a<b$. After the operation, $y$ becomes the leftmost child of $x$. The new rank-difference of $y$ is $b - a$. Observe that the rank-differences of nodes other than $y$ remain unchanged, therefore the potential $\Phi$ decreases by $a$.

\begin{theorem}
The cost of $n$ \dm~operations from an arbitrary initial heap of size $n$, using arbitrary pairing and linking, is $O(n \sqrt{n})$.
\end{theorem}

\begin{proof}
Consider a \dm~in which the root has $k$ children. In the pairing round we perform $\lfloor k/2 \rfloor$ \link~operations. By the earlier observation, the potential $\Phi$ goes down by $a_1 + \cdots + a_{\lfloor k/2 \rfloor}$, where $a_i$ is the rank-difference of the ``winner'' (i.e., smaller node) in the $i$-th \link. The values $a_i$ are distinct (since all respective nodes have different ranks, and the same parent), and $a_i > 0$ for all $i$. Thus, the entire pairing round reduces $\Phi$ by at least $1 + 2 + \dots + \lfloor k/2 \rfloor \geq (k^2-1)/8 \geq (k-1)^2 /8$. The remaining operations can only further decrease the potential.

Let $k_i$ be the number of children of the root before the $i$-th \dm~operation. (This is also our estimate for the cost of the operation.) The total reduction in potential due to this operation is at least $\sum_{i = 1}^{n} {{(k_i-1)}^2/8}$. On the other hand, the total reduction in potential over all $n$ operations is not more than $n^2/2$. It follows that $\sum_i{(k_i-1)^2} \leq 4n^2$. Under this condition, setting $k_1 = \cdots = k_n = 2 \sqrt{n} + 1$ maximizes the total real cost $\sum_i{k_i}$. \end{proof}

In its current form, the potential function is not well suited for analysing operations other than \dm; a single \ins, for example, may increase $\Phi$ by a linear term. 
We sketch a new analysis for sequences of \ins~and \dm~operations that foretells the technique used in \S\,\ref{sec21}. 

Define the potential of a node $x$, denoted by $\phi(x)$, to be $rd(x)$, if $rd(x) \leq \sqrt{n}-1$, and $\sqrt{n} - 1 + rd(x) / \sqrt{n}$ otherwise. Call nodes of the first kind \emph{light}, and nodes of the second kind \emph{heavy}. The potential function $\Phi$ is now defined as the sum of $\phi(x)$ over all nodes $x$.
Observe that for a heap of size $n$, the total potential is $\Phi < 2n\sqrt{n}$.

Consider the pairing round in a \dm~operation. Since there are less than $\sqrt{n}$ light nodes among the children of the root, all but $\sqrt{n}$ of the link operations are between heavy nodes. It is easy to verify that a \link~between two heavy nodes reduces the potential by at least $1$. The amortized $O(\sqrt{n})$ cost of \dm~follows.

Next we bound the amortized cost of an \ins~operation, considering the increase in potential that it may cause. 

First, we have to add the potential of the newly inserted key, which is at most $2\sqrt{n}$. (In case the newly inserted key is the new minimum, it does not contribute to the potential.)

The second, more subtle effect of \ins~is that it may cause a change in the rank-differences of other nodes. A given node $y$ is affected by a newly inserted node $x$ if the rank of $x$ falls between the rank of $y$ and that of its parent $p(y)$. If the affected node $y$ is a light node (before the \ins), then its rank $r(y)$ (before the \ins) can be larger than $r(x)$ by at most $\sqrt{n}-2$ (otherwise, the rank-difference of $y$ would have been too large). There can be at most $\sqrt{n}-1$ such nodes, and the potential of each may go up by $1$. Otherwise, if $y$ is a heavy node, then $\phi(y)$ may go up only by $1/\sqrt{n}$. Overall, the potential $\Phi$ increases by at most $2\sqrt{n}$. The amortized $O(\sqrt{n})$ cost of \ins~follows.

The reader may observe that we ignored the change in potential due to the change in the value of $n$. We can deal with this technicality by keeping $n$ unchanged, as long as it does not get \emph{too far} away from the true number of elements. When that happens, $n$ can be updated by a standard doubling-halving strategy, without affecting the claimed amortized costs. We discuss this issue in more detail in the context of our main result. 

\subsection{The main result\label{sec21}}

\mainclaim*

To prove Theorem~\ref{thm2}, we replace the simple potential function used in \S\,\ref{sec20} by one with a more fine-grained scaling. Again, first we present the tools necessary to analyse the heap in a \emph{sorting mode}, i.e., we compute the cost of $n$ \dm~operations on an arbitrary initial heap of size $n$, then we make the necessary changes to analyse 
both \ins~and \dm.

In the following, for the purpose of analysis, we assume that $n$ is an upper bound on the number of nodes in the heap, not greater than four times the true value. (Except for the beginning, when the heap is empty, and we set $n$ to a small constant value, say $n=4$.) 
At the end, we describe how we update $n$, if, after a certain number of operations, the true value reaches $n$, or falls far below $n$.

Let $q=q(n)$ be the \emph{scaling factor} for $n$, an integer that we optimize later; assume for now that $1<q<n$. The \emph{category} of a (non-root) node $x$, denoted $c(x)$ is defined as $c(x) = \lfloor \log_q{rd(x)} \rfloor$. Observe that $c=c(x)$ is the unique integer for which $q^c \leq rd(x) < q^{c+1}$. For all nodes $x$ we have $0 \leq c(x) < t$, where $t = t(n) = \lfloor \log_q{(n-1)} \rfloor + 1$, and $t\geq2$.

For all non-root nodes $x$, we define the \emph{node potential} of $x$ as

$$\phi(x) = \frac{rd(x)-q^c}{q^{c} - q^{c-1}} + c \cdot q, \text{~where~} c=c(x).$$

For the root $r$, we let $\phi(r) = 0$. The total node potential is $\Phi_N = \sum {\phi(x)}$, where the sum ranges over all nodes $x$. Some observations about the node potential are in order.

\begin{lemma}\label{obs}
With the previous definitions, we have:
\begin{enumerate}[(i)]
\item If $rd(y)=rd(x)+1$, and $rd(x) \geq 1$, then $\phi(y)=\phi(x)+1/(q^{c}-q^{c-1})$, where $c=c(x)$.

\item For every node $x$, it holds that $ 0 \leq \phi(x) \leq t \cdot q$, and therefore $\Phi_N \leq n \cdot t \cdot q$.

\item If two nodes of the same category are linked, then $\Phi_N$ decreases by at least 1.
\end{enumerate}
\end{lemma}
\begin{proof}
\begin{enumerate}[(i)]
\item Suppose $c(y) = c(x) = c$. Then, $\phi(y)-\phi(x) = \left(rd(y)-rd(x)\right)/(q^{c}-q^{c-1}) = 1/(q^{c}-q^{c-1})$.
Otherwise, suppose $c(y) = c(x) + 1$. Let $c = c(x)$. Then, $rd(y) = q^{c+1}$, and $rd(x) = q^{c+1}-1$. Thus, $\phi(y)-\phi(x) = (c+1) \cdot q - (q^{c+1}-1-q^{c})/(q^{c} - q^{c-1}) - c \cdot q = 1/(q^{c} - q^{c-1})$.

\item From (i) it follows that $\phi(x)$ is maximal if $rd(x) = n-1$. Then, $c(x) = t-1$, and $\phi(x) \leq (q^{t}-q^{t-1})/(q^{t-1} - q^{t-2}) + (t-1) \cdot q = t \cdot q$.

\item

Let us define the function $f(\cdot)$ that maps $rd(x)$ to $\phi(x)$ for all $x$. It is easy to verify that $f(\cdot)$ is strictly increasing for $rd(x) \geq 1$.

Suppose that $x<y$, and $c(x) = c(y) = c$. Then, $q^c \leq rd(x) < rd(y) < q^{c+1}$. Only the potential of $y$ changes after the link operation, from $f(rd(y))$ to $f(rd(y)-rd(x))$.

We want to show $f(rd(y)) - f(rd(y)-rd(x)) \geq 1$. This quantity is minimized if $rd(x) = q^c$. By (i), $f(rd(y)-q^c) \leq f(rd(y)) - q^c/(q^c-q^{c-1}) \leq f(rd(y))-1$. The result follows. \qedhere 
\end{enumerate}
\end{proof}

\vspace{-0.05in}
Suppose that $x_1, \dots, x_k$ are children of the same node in the heap, indexed in left-to-right order. A subset $\{x_i, x_{i+1}, \dots, x_j\}$ of these nodes, for $1 \leq i < j \leq k$, is referred to as a \emph{contiguous group} of siblings. Within the heap, we consider certain contiguous groups of siblings to be \emph{in a box}; this is only for the purpose of the analysis, with no effect on the implementation.  

We define a second kind of potential to capture the current state of the boxes.  Throughout the lifetime of the heap we maintain a number of invariants about the boxes, as follows.
\begin{enumerate}[(A)]
\item The size of a box (i.e.,  the number of nodes in a box) is of the form $2^b - 1$, for some $2 \leq b \leq t$.
\item Boxes are pairwise disjoint. 
\item A box of size $2^b-1$ can only contain nodes 
of category $b-2$ or smaller.
\end{enumerate}

\vspace{-0.05in}
The \emph{box potential}, denoted $\Phi_B$ is a sum over all boxes of $(b-t-1)/t$, where $2^{b}-1$ is the size of the box. In particular, the largest possible box of size $2^{t}-1$ contributes $-1/t$ to the box potential, and the smallest possible box of size $3$ contributes $(1/t-1)$. Observe that $-n < \Phi_B \leq 0$. At the beginning of the operations there are no boxes, therefore $\Phi_B = 0$.

The \emph{total potential} $\Phi$ combines the node potential and the box potential, as $\Phi = \Phi_N + \Phi_B$.

We are ready to analyse the individual operations and their effect on the potential.

\vspace{-0.15in}
\subparagraph*{Delete-min only.} Again, assume that the \dm~operation performs the pairing and accumulation rounds first, and deleting the root afterwards.  We need to account for the change of node and box potential in all steps.

Consider an operation $\link(x,y)$, and assume w.l.o.g.\ that $x < y$, and therefore $x$ becomes the parent of $y$. We say that the link is a \emph{good link} if one of the following holds:
\begin{enumerate}[(1)]
\item before the link $x$ and $y$ were of the same category,
\item the link occurs in the accumulation round, is not of type (1), and $y$ is the left-to-right minimum before the link (being replaced by $x$ in this role).
\end{enumerate}

\vspace{-0.05in}
At a high level, our strategy is to show that for every $t \cdot 2^t$ link operations, one good link occurs. A good link of the type specified in case (1) decreases the node potential by at least $1$. Therefore, by Lemma~\ref{obs}(ii), there are at most $n \cdot t \cdot q$ such links. (In our current setting no action increases the node potential.) A good link of the type specified in case (2) can occur at most $t-1$ times during a \dm. To see that there are no more than $t-1$ type-(2) good links in a \dm, observe that such an event necessarily leads to a decrease in category of the current left-to-right minimum, and the number of possible categories is $t$.

The link operations that we charge to a single good link may be scattered through multiple \dm~operations. The purpose of the boxes is to keep track of link operations that we have not yet accounted for.

Let $k$ denote the number of children of the deleted root, i.e.,  the real cost of the operation, up to a constant factor, and let us look in turn at the pairing round, accumulation round, and the deletion of the root. \\

\noindent\textit{Pairing round.} Consider the links involving nodes from a box of size $2^b-1$. If at least one of these is a good link of type (1), we simply remove the box. The result is an increase in box potential by $(t-b+1)/t$ due to the deletion of the box, and a decrease in node potential by at least $1$ due to the good link. This amounts to a decrease in total potential of $(b-1)/t \geq 1/t$. 
Observe that if the box size before the operation is $2^2-1 = 3$ (i.e.,  minimal), then the link operation within the box is \emph{always} a good link. This is because, by invariant (C), all nodes in the box are of category $0$.

Suppose that we process a box of size $2^b-1$ for $b >2$, and no good link occurs within the box. Then, the box continues to exist around the winners of the link operations where both nodes were in the same box, i.e., the \emph{losers} of the link operations, and possibly the nodes ``on the margin'' leave the box. Observe that the size of the box goes from $2^b-1$ to $2^{b-1}-1$. Due to the decrease of the box size, we have a decrease in box potential by $1/t$, and consequently a decrease by $1/t$ in total potential. (See Figure~\ref{fig2} for illustration.)

\begin{figure}
	\begin{center}
\includegraphics[width=0.99\textwidth]{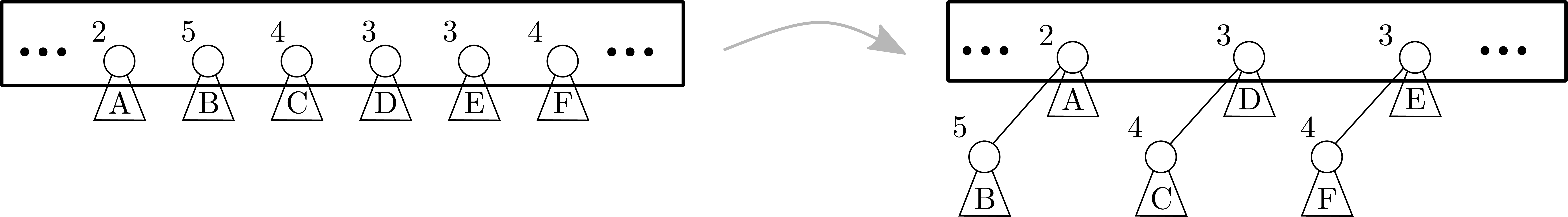}
	\end{center}
\caption{The evolution of a box. Numbers indicate categories of nodes. The losers of the comparisons drop out of the box.
\label{fig2}}
\end{figure}

Now suppose we execute $T = 2^t-1$ consecutive non-good link operations without encountering a box. Then, we form a box around the $2^t-1$ winners of these links. Due to the newly created box of maximum size, the total potential decreases by $1/t$.

Consider a sequence of $T$ consecutive links. Then, either (i) one of the links is good, or (ii) we create a new box containing the $T$ winners, or (iii) we encounter an existing box. In the latter case, we process the box until its end. Thus, in all cases, for at most $2T-1$ consecutive link operations, we achieve a saving in potential of at least $1/t$. 
This yields (for the entire pairing round) a decrease in potential of at least 
$$ \Bigl\lfloor \frac{\lfloor k/2 \rfloor}{2 \cdot 2^t - 3} \Bigr\rfloor \cdot \frac{1}{t} \geq \frac{k}{4 \cdot t\cdot 2^t} - \frac{2}{t}.$$

We need to argue that in the cases described above, the box invariants are maintained. Condition (A) clearly holds in every case. When we create a new box, we only use nodes that are not in a box, therefore condition (B) holds. We never add new nodes to an existing box, thus (B) continues to hold. We next discuss the validity of invariant (C), which is the crucial ingredient of the proof.

When a new box of size $2^t-1$ is created, we claim that there cannot be any node of category $t-1$ within the box. This is because, if any node in the box was of category $t-1$, its ``winning'' of the linking would have happened against another node of category $t-1$, a type-(1) good link that contradicts our condition for creating the box.

Similarly, when shrinking an existing box from size $2^b - 1$ to $2^{b-1} - 1$, if invariant (C) was true before the operation (i.e.,  there were only nodes of category $0,\dots,b-2$ in the box), then the nodes in the box after the operation can only be of category $0,\dots,b-3$. A node of category $b-2$ could only have ``won'' against another node of category $b-2$, a type-(1) good link that contradicts our condition for maintaining the box. This establishes the invariants. \\

\noindent\textit{Accumulation round.} By Lemma~\ref{obs}(i), the node potential can only further decrease. We need to argue that the box-structure is not destroyed by this round.

Recall that we start with the leftmost node (the current minimum), and keep linking against the nodes from left to right. As long as there is no type-(2) good link, the nodes that we encounter left-to-right become the children of the current minimum in right-to-left order. The box invariants are clearly not affected by this reversal of order, the box, together with its parent node simply moves further down the tree. (Categories of nodes in the box can only decrease.)

If there is a change in the minimum, i.e.,  a type-(2) event, while we are processing the nodes in a box, then we delete the box. By deleting the box, we may increase the box potential by $1 - 1/t < 1$. This is the only kind of potential-change we need to consider in this round, yielding an increase in total potential of at most $t-1$ (since there are at most $t-1$ type-(2) events). \\

\noindent\textit{Deleting the root.} The actual deletion of the root leaves the node potential unchanged, since 
at that time the root has only one child, its successor, which has rank-difference $1$, and consequently, node potential $0$.
The box potential is unaffected, since the root cannot be in a box. \\

\vspace{-0.1in}
To summarize, if the actual cost of \dm~is $k$, then, collecting the terms from the different rounds, the decrease in potential is at least $k/(4t \cdot 2^t) - 2/t - (t-1) \geq k/(4t \cdot 2^t) - t$.

Denoting the actual costs of the $n$ \dm~oparations by $k_1,\dots,k_n$, for the total potential decrease $\Delta \Phi$ we have: 
\vspace{-0.1in}
$$\Delta \Phi \geq \left( \frac{1}{4t \cdot 2^t}\sum_{i=1}^n{k_i} \right) - n \cdot t.$$

From Lemma~\ref{obs}(i) it follows that $\Phi \leq n \cdot t \cdot q$, for an arbitrary heap. For the empty heap, we have $\Phi = 0$. Thus, $\Delta \Phi \leq n \cdot t \cdot q$. It follows that the total cost is $\sum{k_i} \leq (n\cdot t \cdot q + n \cdot t) \cdot (4t \cdot 2^t)$.

Let us now choose the scaling factor used in the analysis. Recall that $t = \lfloor \log_q{(n-1)} \rfloor$. Fixing $t = \lfloor \sqrt{\log{n}} \rfloor$, we have $q = O(n^{1/\sqrt{\log{n}}}) = O(2^{\sqrt{\log{n}}})$. Thus, we obtain that the amortized cost of \dm~is $O(\log{n} \cdot 4^{\sqrt{\log{n}}})$.

\vspace{-0.15in}
\subparagraph*{Insert and Delete-min.}

Next, we consider sequences of \ins~and \dm~operations, arbitrarily intermixed. The amortised cost of an operation is defined as the true cost \emph{plus} the increase in potential due to the operation. 

The analysis of \dm~is the same as before, yielding a decrease in total potential due to a \dm~of at least $k/(4t \cdot 2^t) - t$, where $k$ is the actual cost of the operation.

The true cost of \ins~is $O(1)$, but we also need to consider the increase in potential. The box potential is not affected by \ins, as we do not create any new boxes and neither the root, nor the newly inserted node are in a box.

When a key $x$ is inserted, its first contribution is its own node potential $\phi(x)$, in case it becomes the child of the root. (If $x$ is smaller than the root, then, after the linking, both $x$ and the old root have node potential zero.) 
By Lemma~\ref{obs}(ii), the resulting increase in the node potential is at most $t \cdot q$.

As in \S\,\ref{sec20}, we need to deal with the fact that an \ins~operation may change the ranks of existing nodes, possibly increasing their rank-difference. We could proceed with an argument similar to the one in \S\,\ref{sec20}, estimating the change in potential for nodes in each category. The argument gets complicated, however, because the increase in rank-difference may also cause an increase in category, invalidating the box invariants.

To avoid\footnote{We observe that our problem of maintaining the ranks under insertions is reminiscent of the well-studied \emph{ordered list maintenance} problem, for which efficient (although quite involved) solutions exist. Luckily, we can get away with a simpler solution.} the issue of updating ranks during an \ins~operation, we redefine \emph{rank} to take into account not just the items currently in the heap, but also the items that will be inserted into the heap in the future. (Since we use ranks only in the analysis, we can assume that future operations are known to us.) In this way, an \ins~operation has no effect on ranks, rank-differences, or categories of existing items, since the rank of the newly inserted item has already been taken into consideration. Therefore, there is no further change in potential that we need to consider.

There remains the issue, that we cannot afford to look \emph{too far} into the future, as we want the value $n$ to be, at all times, close to the current size of the heap. Therefore, we split the operations into epochs, and keep the value of $n$ fixed throughout an epoch. If the heap is empty, e.g., in the beginning, we let $n=4$. 

An epoch ends when the true size of the heap increases to $n$ (after an \ins), or decreases to $\lfloor n/4 \rfloor$ (after a \dm), or if, within the epoch, $n$ distinct keys have already been encountered (including those that were in the heap at the beginning of the epoch). In all three cases, we reset $n$ to be twice the true size of the heap, remove all boxes, and start a new epoch. 
(We stress that epochs and boxes are used only in the analysis, with no effect on the actual operation of the data structure.)
As mentioned, the ranks of the items are computed with respect to \emph{all items} encountered during an epoch, they are therefore, fixed within the epoch. Clearly, all ranks and rank-differences are still upper bounded by $n$.

We make four observations, all of which are easy to verify based on the preceding discussion: (1) The true size of the heap is between $\lfloor n/4 \rfloor$ and $n$. (2) Within a finished epoch there must have been at least $\lceil n/4 \rceil$ operations. (3) The increase in node potential due to the resetting of $n$ is not more than $n \cdot t \cdot q$. (4) Since the boxes are disjoint, there are less than $n$ boxes that we are removing, increasing the total potential by at most $n$. 

We distribute the increase in potential of at most $n \cdot t \cdot q + n$ among the $\lceil n/4 \rceil$ operations in the finished epoch, obtaining that the \emph{increase} in potential during an \ins~is at most $5 \cdot t \cdot q + 4$, and the increase in potential during a \dm~is at most $t - k/(4t \cdot 2^t)  + 4 \cdot t \cdot q + 4$. Scaling the potential by $t \cdot 2^t$, we obtain that the amortized cost of both \ins~and \dm~is $O(t^2 \cdot q \cdot 2^t)$. 

Again, choosing $t = \lfloor \sqrt{\log{n}} \rfloor$, we obtain $q = O(n^{1/\sqrt{\log{n}}}) = O(2^{\sqrt{\log{n}}})$, and the amortized costs of $O(\log{n} \cdot 4^{\sqrt{\log{n}}})$ follow. 

\subparagraph*{Meld.} A \meld~operation is similar to an \ins, in that only one node changes its rank-difference, and thereby its node potential (the root with the larger key). Furthermore, \meld~also leaves the box potential unaffected, so its analysis goes through similarly to the analysis of \ins, yielding the same amortized cost for \meld~as for \ins~and \dm. Observe, however, that if we also include \meld~operations in a sequence of operations, then the value $n$ that appears in the cost no longer denotes just the size of the affected heap, it is instead, the number of items in \emph{all heaps} that we are currently working with.

\newpage
\subparagraph*{Other variants.}
The presented analysis is fairly general. It is not difficult to adapt it to the \emph{standard} and \emph{multipass} variants of pairing heaps, yielding similar upper bounds. For the standard variant the analysis is essentially the same as for the forward variant. For multipass, the analysis becomes simpler, since boxes need to be maintained only during individual \dm~operations. However, for these variants, stronger upper bounds are already known, as discussed in \S\,1.

Recall that during a \link~operation, the larger of the two items is linked as the leftmost child of the smaller.
The reader may observe that the different behaviors of the standard and forward variants depend on this particular way of implementing \link. If we were to link by making the larger item the \emph{rightmost} child of the smaller, then the situation would reverse, with the forward variant being easy, and the standard variant hard to analyse.

We may modify the implementation of \link, such as to link \emph{arbitrarily} as the leftmost or rightmost child (i.e., deciding independently for every \link~operation whether to link at the left or at the right). Our analysis also extends to this more general class of pairing heap algorithms (that use the modified \link~implementation) with minimal changes. The only difficulty that arises in the new setting is that during the accumulation round, an existing box may split into two, one part going to the left, the other to the right side of the current minimum. To account for this loss, we need to start with a larger initial box ($4^t$ instead of $2^t$). The resulting bounds are still of the form $2^{O(\sqrt{\log{n}})}$.

Improving (significantly) the bounds presented in the paper and extending the analysis to other operations should be possible but will likely require new ideas.

\newpage







\bibliography{submit}

\begin{thebibliography}{10}

\bibitem{avl}
G.~M. Adelson-Velski\u{\i} and E.~M. Landis.
\newblock An algorithm for organization of information.
\newblock {\em Dokl. Akad. Nauk SSSR}, 146:263--266, 1962.

\bibitem{Bayer1972}
Rudolf Bayer.
\newblock Symmetric binary b-trees: Data structure and maintenance algorithms.
\newblock {\em Acta Informatica}, 1(4):290--306, Dec 1972.

\bibitem{ESA15}
Parinya Chalermsook, Mayank Goswami, L{\'{a}}szl{\'{o}} Kozma, Kurt Mehlhorn,
  and Thatchaphol Saranurak.
\newblock Self-adjusting binary search trees: What makes them tick?
\newblock In {\em {ESA} 2015}, pages 300--312, 2015.

\bibitem{Quake}
Timothy~M. Chan.
\newblock {\em Quake Heaps: A Simple Alternative to Fibonacci Heaps}, pages
  27--32.
\newblock Springer Berlin Heidelberg, Berlin, Heidelberg, 2013.

\bibitem{Cormen}
Thomas~H. Cormen, Clifford Stein, Ronald~L. Rivest, and Charles~E. Leiserson.
\newblock {\em Introduction to Algorithms}.
\newblock McGraw-Hill Higher Education, 2nd edition, 2001.

\bibitem{Elmasry1}
Amr Elmasry.
\newblock Pairing heaps with \emph{O}(log log \emph{n}) decrease cost.
\newblock In {\em Proceedings of the Twentieth Annual {ACM-SIAM} Symposium on
  Discrete Algorithms, {SODA} 2009, New York, NY, USA, January 4-6, 2009},
  pages 471--476, 2009.

\bibitem{Violation}
Amr Elmasry.
\newblock The violation heap: a relaxed fibonacci-like heap.
\newblock {\em Discrete Math., Alg. and Appl.}, 2(4):493--504, 2010.

\bibitem{FredmanLB}
Michael~L. Fredman.
\newblock On the efficiency of pairing heaps and related data structures.
\newblock {\em J. {ACM}}, 46(4):473--501, 1999.

\bibitem{pairing}
Michael~L. Fredman, Robert Sedgewick, Daniel~Dominic Sleator, and Robert~Endre
  Tarjan.
\newblock The pairing heap: {A} new form of self-adjusting heap.
\newblock {\em Algorithmica}, 1(1):111--129, 1986.

\bibitem{fibonacci}
Michael~L. Fredman and Robert~Endre Tarjan.
\newblock Fibonacci heaps and their uses in improved network optimization
  algorithms.
\newblock In {\em 25th Annual Symposium on Foundations of Computer Science,
  West Palm Beach, Florida, USA, 24-26 October 1984}, pages 338--346, 1984.

\bibitem{Guibas78}
Leonidas~J. Guibas and Robert Sedgewick.
\newblock A dichromatic framework for balanced trees.
\newblock In {\em 19th Annual Symposium on Foundations of Computer Science, Ann
  Arbor, Michigan, USA, 16-18 October 1978}, pages 8--21, 1978.

\bibitem{RankPairing}
Bernhard Haeupler, Siddhartha Sen, and Robert~Endre Tarjan.
\newblock Rank-pairing heaps.
\newblock {\em {SIAM} J. Comput.}, 40(6):1463--1485, 2011.

\bibitem{Hollow}
Thomas~Dueholm Hansen, Haim Kaplan, Robert~Endre Tarjan, and Uri Zwick.
\newblock Hollow heaps.
\newblock In {\em Automata, Languages, and Programming - 42nd International
  Colloquium, {ICALP} 2015, Kyoto, Japan, July 6-10, 2015, Proceedings, Part
  {I}}, pages 689--700, 2015.

\bibitem{IaconoUB}
John Iacono.
\newblock Improved upper bounds for pairing heaps.
\newblock In {\em Algorithm Theory - {SWAT} 2000, 7th Scandinavian Workshop on
  Algorithm Theory, Bergen, Norway, July 5-7, 2000, Proceedings}, pages 32--45,
  2000.

\bibitem{in_pursuit}
John Iacono.
\newblock In pursuit of the dynamic optimality conjecture.
\newblock In {\em Space-Efficient Data Structures, Streams, and Algorithms},
  volume 8066 of {\em Lecture Notes in Computer Science}, pages 236--250.
  Springer Berlin Heidelberg, 2013.

\bibitem{IaconoOzkan2}
John Iacono and {\"{O}}zg{\"{u}}r {\"{O}}zkan.
\newblock A tight lower bound for decrease-key in the pure heap model.
\newblock {\em CoRR}, abs/1407.6665, 2014.

\bibitem{IaconoOzkan}
John Iacono and {\"{O}}zg{\"{u}}r {\"{O}}zkan.
\newblock Why some heaps support constant-amortized-time decrease-key
  operations, and others do not.
\newblock In {\em Automata, Languages, and Programming - 41st International
  Colloquium, {ICALP} 2014, Copenhagen, Denmark, July 8-11, 2014, Proceedings,
  Part {I}}, pages 637--649, 2014.

\bibitem{Iacono16}
John Iacono and Mark Yagnatinsky.
\newblock {\em A Linear Potential Function for Pairing Heaps}, pages 489--504.
\newblock Springer International Publishing, Cham, 2016.

\bibitem{ThinThick}
Haim Kaplan and Robert~Endre Tarjan.
\newblock Thin heaps, thick heaps.
\newblock {\em {ACM} Trans. Algorithms}, 4(1):3:1--3:14, 2008.

\bibitem{Katriel}
Irit Katriel, Peter Sanders, and Jesper~Larsson Tr{\"{a}}ff.
\newblock A practical minimum spanning tree algorithm using the cycle property.
\newblock In {\em Algorithms - {ESA} 2003, 11th Annual European Symposium,
  Budapest, Hungary, September 16-19, 2003, Proceedings}, pages 679--690, 2003.

\bibitem{Knuth}
Donald~E. Knuth.
\newblock {\em The Art of Computer Programming, Volume 1 (3rd Ed.): Fundamental
  Algorithms}.
\newblock Addison Wesley Longman Publishing Co., Inc., Redwood City, CA, USA,
  1997.

\bibitem{Tarjan14}
Daniel~H. Larkin, Siddhartha Sen, and Robert~Endre Tarjan.
\newblock A back-to-basics empirical study of priority queues.
\newblock In {\em 2014 Proceedings of the Sixteenth Workshop on Algorithm
  Engineering and Experiments, {ALENEX} 2014, Portland, Oregon, USA, January 5,
  2014}, pages 61--72, 2014.

\bibitem{Moret}
Bernard M.~E. Moret and Henry~D. Shapiro.
\newblock {\em An empirical analysis of algorithms for constructing a minimum
  spanning tree}, pages 400--411.
\newblock Springer Berlin Heidelberg, Berlin, Heidelberg, 1991.

\bibitem{Pettie}
Seth Pettie.
\newblock Towards a final analysis of pairing heaps.
\newblock In {\em 46th Annual {IEEE} Symposium on Foundations of Computer
  Science {(FOCS} 2005), 23-25 October 2005, Pittsburgh, PA, USA, Proceedings},
  pages 174--183, 2005.

\bibitem{pettie_slides}
Seth Pettie.
\newblock Thirteen ways of looking at a splay tree.
\newblock Unpublished lecture, 2014.

\bibitem{ST85}
Daniel~Dominic Sleator and Robert~Endre Tarjan.
\newblock Self-adjusting binary search trees.
\newblock {\em J. {ACM}}, 32(3):652--686, 1985.

\bibitem{Stasko}
John~T. Stasko and Jeffrey~Scott Vitter.
\newblock Pairing heaps: Experiments and analysis.
\newblock {\em Commun. ACM}, 30(3):234--249, March 1987.

\bibitem{Subramanian96}
Ashok Subramanian.
\newblock An explanation of splaying.
\newblock {\em J. Algorithms}, 20(3):512--525, 1996.

\bibitem{Vuillemin78}
Jean Vuillemin.
\newblock A data structure for manipulating priority queues.
\newblock {\em Commun. {ACM}}, 21(4):309--315, 1978.

\end{thebibliography}

\end{document}